\documentclass[pra,reprint]{revtex4-1}
\usepackage{graphicx,epstopdf}

\usepackage{amsmath,amsfonts,amsthm}
\usepackage{enumitem}
\usepackage{scalerel}

\usepackage{subcaption}

\newcommand{\mc}{\mathcal}

\newcommand{\ket}[1]{|#1\rangle}

\newcommand{\bkt}[2]{\langle#1|#2\rangle}
\newcommand{\nrm}[1]{\left\| #1 \right\|}

\newtheorem{corollary}{Corollary}
\newtheorem{theorem}{Theorem}

\begin{document}

% Use the \preprint command to place your local institutional report
% number in the upper righthand corner of the title page in preprint mode.
% Multiple \preprint commands are allowed.
% Use the 'preprintnumbers' class option to override journal defaults
% to display numbers if necessary
%\preprint{}

%Title of paper
\title{A $\psi$-Ontology Result without the Cartesian Product Assumption}

% repeat the \author .. \affiliation  etc. as needed
% \email, \thanks, \homepage, \altaffiliation all apply to the current
% author. Explanatory text should go in the []'s, actual e-mail
% address or url should go in the {}'s for \email and \homepage.
% Please use the appropriate macro foreach each type of information

% \affiliation command applies to all authors since the last
% \affiliation command. The \affiliation command should follow the
% other information
% \affiliation can be followed by \email, \homepage, \thanks as well.
\author{Wayne C. Myrvold}
\affiliation{Department of Philosophy, The University of Western Ontario}
\email{wmyrvold@uwo.ca}
%\homepage[]{Your web page}
%\thanks{}
%\altaffiliation{}

%Collaboration name if desired (requires use of superscriptaddress
%option in \documentclass). \noaffiliation is required (may also be
%used with the \author command).
%\collaboration can be followed by \email, \homepage, \thanks as well.
%\collaboration{}
%\noaffiliation

\date{\today}

\begin{abstract}
We introduce a weakening of the Preparation Independence Postulate of Pusey, Barrett, and Rudolph that does not presuppose that the space of ontic states resulting from a product state preparation can be represented by the Cartesian product of subsystem state spaces.  On the basis of this weakened assumption, it is shown that, in any model that  reproduces the quantum probabilities, any pair of pure quantum states $\ket{\psi}$, $\ket{\phi}$ with $\bkt{\phi}{\psi} \leq 1/\sqrt{2}$  must be ontologically distinct.
\end{abstract}

% insert suggested PACS numbers in braces on next line
%\pacs{}
% insert suggested keywords - APS authors don't need to do this
%\keywords{}

%\maketitle must follow title, authors, abstract, \pacs, and \keywords
\maketitle

% body of paper here - Use proper section commands
% References should be

\section{Introduction} The Pusey-Barrett-Rudolph (PBR) theorem \cite{PBR} demonstrates, against the background of the ontological models framework \cite{SpekkensHarrigan}, that from an independence assumption regarding product-state preparations, the Preparation Independence Postulate, it follows   that any pair of distinct pure quantum states are ontologically distinct---that is, for any theory that reproduces the quantum probabilities for outcomes of experiments, the probability distributions over the ontic state space corresponding to preparations of the two states have null overlap.

The context of the theorem is nonrelativistic quantum mechanics, which is known not to be a fundamental theory, but a nonrelativistic limit of a more fundamental theory, quantum field theory, which is itself thought to be a low-energy limit of some more fundamental theory. This raises the question of the relevance of such results to the ontology of the actual world. One attitude that has been taken is that such results concern a hypothetical world in which quantum mechanics is exactly correct \cite{LeiferPsiOnt}. Another attitude, which is the one that is adopted in this article, is that one can and should seek to draw conclusions about the actual world from theorems of this sort. This requires one to seek to formulate theorems on the basis of assumptions that might reasonably be expected to hold of successor theories to quantum mechanics.  In particular, the theorems should not rely on assumptions that are generically violated by quantum field theories.

The key assumption of the PBR theorem is the Preparation Independence Postulate (PIP).  This is meant to apply to independent pure-state preparations performed on two or more quantum systems. The postulate is the conjunction of two assumptions. On the terminology of ref. \cite{LeiferPsiOnt}, these are the \emph{Cartesian Product Assumption} (CPA), that the state space of the joint system may be represented as the Cartesian product of state spaces of the subsystems, and the \emph{No Correlation Assumption} (NCA), that the joint probability distributions corresponding to the various possible joint preparations are products of distributions corresponding to preparations on the individual subsystems.

The CPA is violated in a relativistic quantum field theory. In such theories, entanglement between spacelike separated regions is generic, and independent preparations at spacelike separation cannot be assumed to remove all entanglement, though they may produce an arbitrarily close approximation to a product state (see ref. \cite{EntangOpenSystems} for discussion).  Any theorem that can be taken  to be robust under the transition from nonrelativistic quantum mechanics to a relativistic quantum field theory, and beyond, must not rely on the Cartesian Product Assumption.

In what follows, we adopt the framework of the PBR theorem but drop the CPA.  As the NCA cannot even be formulated in the absence of the CPA, a substitute is required that does not presuppose the CPA.  The independence assumption adopted is what we will call the \emph{Preparation Uninformativeness Condition} (PUC):  when choices of preparation procedures to be performed on two or more systems are made independently, then, given a complete specification of the resulting ontic state,  information about which preparation was performed on one system is not informative about which preparation was performed on the others.  This, of course, is entailed by the PIP, but it is strictly weaker than it.  As is shown in the Appendix, even if the NCA is assumed, the PUC does not entail the NCA.

The conclusion we will arrive at is that any pair of pure quantum states $\ket{\psi}$, $\ket{\phi}$ with $|\bkt{\phi}{\psi}| \leq 1/\sqrt{2}$  must be ontologically distinct. The method used in ref. \cite{PBR} to extend a result of this sort to arbitrary distinct states makes essential use of the Preparation Independence Postulate.  There does not seem to be any obvious way to achieve the same result using only the PUC; we leave it as an open question whether this can be done.

\section{The ontological models framework}  We adopt the ontological models framework of ref. \cite{SpekkensHarrigan}.  In this framework, with any physical system is associated a measurable space $\langle \Lambda, \mc{S} \rangle$, where $\Lambda$ is the space of possible ontic states of the system, and $\mc{S}$ is a $\sigma$-algebra of subsets of $\Lambda$, which are the subsets to which probabilities will be assigned.   With any preparation procedure $\psi$ is associated a probability measure $P_\psi$ on  $\langle \Lambda, \mc{S} \rangle$.    It is assumed that, for any experiment $E$ that can be performed on the system, with a finite set $K_E$ of outcomes, there exists, for each $k \in K_E$, a measurable function $p^E_k:\Lambda \rightarrow [0, 1]$, such that for any $\lambda \in \Lambda$, $p^E_k(\lambda)$ is the probability, in state $\lambda$, that the experiment $E$ yields the result $k$.  These must satisfy, for all $\lambda \in \Lambda$,
\begin{equation}
\sum_{k \in K_E} p^E_k(\lambda) = 1.
\end{equation}
The probability of result $k$, given preparation $\psi$, is then
\begin{equation}
P_\psi^E(k) = \langle \, p^E_k \, \rangle_\psi,
\end{equation}
where $\langle \; \cdot \; \rangle_\psi$ denotes  expectation value with respect to the probability measure $P_\psi$.

We will say that two  quantum states $\ket{\psi}$, $\ket{\phi}$ are \emph{ontologically distinct} if and only if, for any procedures that prepare the two states, the corresponding probability distributions have disjoint supports.  We will also want to consider states that are close to being ontologically distinct, and for this we need a way of quantifying the difference between two probability distributions. One such way is afforded by the \emph{statistical distance}, also known as the \emph{total variation distance}:
\begin{equation}
\delta(P, Q) = \sup_{A \in \mc{S}} \left|P(A) - Q(A) \right|.
\end{equation}
Its value ranges between 0, when $P = Q$, and $1$, when $P$ and $Q$ have disjoint supports. If $P$ and $Q$ have densities  $p$, $q$, with respect to some measure $\mu$, the statistical distance can be written in terms of these densities.
\begin{equation}
\delta(P, Q) = \frac{1}{2} \int_\Lambda \left|p - q \right| \, d\mu  = 1 - \int_\Lambda \min(p, q) \, d\mu.
\end{equation}
Define the classical overlap of the two distributions as
\begin{equation}
\omega(P, Q) = 1 - \delta(P, Q) = \int_\Lambda \min(p, q) \, d\mu.
\end{equation}

We will also be concerned with the \emph{Hellinger distance}, which is proportional to the $\mc{L}^2$ distance between $\sqrt{p}$ and $\sqrt{q}$.
\begin{multline}
H(P, Q) = \frac{1}{\sqrt{2}} \nrm{\sqrt{p} - \sqrt{q}}_2 \\ = \left(\frac{1}{2} \int_\Lambda \left(\sqrt{p} - \sqrt{q} \right)^2  d \mu \right)^{1/2}.
\end{multline}
This satisfies
\begin{equation}
H(P, Q)^2 = 1 - \int_\Lambda \sqrt{p \, q} \, d\mu.
\end{equation}
The \emph{fidelity} $F(P,Q)$, also known as the \emph{Battacharyya coefficient},  is defined as
\begin{equation}
F(P,Q) = 1 - H(P, Q)^2  =  \int_\Lambda \sqrt{p \, q} \, d\mu.
\end{equation}
The classical overlap $\omega$ and the fidelity $F$ satisfy the inequalities
\begin{equation}\label{lapineq}
\omega(P, Q) \leq F(P, Q) \leq \left(\int_\Lambda p \, q \, d\mu \right)^{1/2}.
\end{equation}
This gives,
\begin{equation}
\delta(P, Q) \geq H(P, Q)^2 \geq 1 - \left(\int_\Lambda p q \, d\mu \right)^{1/2}.
\end{equation}

For our $\psi$-ontology result, we will not presume that the state space of a composite system is the Cartesian product of the state spaces of subsystems.  It does not follow from this that one cannot attribute ontic states to the components of a composite system.  Given a system $\Sigma$, with state space $\Lambda$,  composed of subsystems $\{ \sigma_i \}$, we assume subsystem state spaces $\{ \Lambda_i \}$.  Presumably, the state of $\Sigma$ uniquely determines the states of the subsystems. This means that, for each $i$,  there is a function $f_i : \Lambda \rightarrow \Lambda_i$ such that $f_i(\lambda)$ is the ontic state of $\sigma_i$ when the ontic state of $\Sigma$ is $\lambda$. One possibility is, of course, that these are projections from a Cartesian product to its components.  Another possibility would be for the ontic state space of $\Sigma$ to consist of quantum states and for the functions $f_i$ to yield reduced states obtained by tracing out all components except $\sigma_i$. In what follows, nothing is presupposed about the relation between the  state space $\Lambda$ and the subsystem spaces $\Lambda_i$, which will play no role in the arguments.

\section{PBR's example} We will utilize the example employed by PBR.  Consider two systems, located in regions $A$ and $B$ (which we will use to denote both the regions and the systems located therein).  Let $\{\ket{0}_A, \ket{1}_A\}$ be a pair of orthogonal states of $A$, and similarly for $\{\ket{0}_B, \ket{1}_B\}$.  Define
\begin{equation}
\begin{array}{l}
\ket{+}_{A,B} = \frac{1}{\sqrt{2}}\left(\ket{0}_{A,B} + \ket{1}_{A,B} \right)
\\ \\
\ket{-}_{A,B} = \frac{1}{\sqrt{2}}\left(\ket{0}_{A,B} - \ket{1}_{A,B} \right)
\end{array}
\end{equation}
Each of the systems is prepared, independently, in its $\ket{0}$ or $\ket{+}$ state.  This yields four possible preparations of the joint system: $\{ \ket{0}_A \ket{0}_B, \ket{0}_A \ket{+}_B, \ket{0}_A \ket{+}_B, \ket{+}_A \ket{+}_B\}$.  It is assumed that there are probability distributions $\{P_{00}, P_{0+}, P_{+0}, P_{++} \}$ on the ontic state space of the joint system, corresponding to these four preparations.

Consider the following orthonormal basis for the Hilbert space of the joint system:
\begin{equation}
\begin{array}{l}
\ket{\xi_1} = \frac{1}{\sqrt{2}} \left(\ket{0}_A\ket{1}_B + \ket{1}_A\ket{0}_B  \right)
\\ \\
\ket{\xi_2} = \frac{1}{\sqrt{2}} \left(\ket{0}_A \ket{-}_B + \ket{1}_A \ket{+}_B  \right)
\\ \\
\ket{\xi_3} =  \frac{1}{\sqrt{2}} \left(\ket{+}_A \ket{1}_B + \ket{-}_A \ket{0}_B  \right)
\\ \\
\ket{\xi_4} =  \frac{1}{\sqrt{2}} \left(\ket{+}_A \ket{-}_B + \ket{-}_A \ket{+}_B  \right).
\end{array}
\end{equation}
Each of these states is orthogonal to one of the preparation states. If an experiment is performed with these states as the possible results, the condition that quantum probabilities be recovered entails that, on distribution $P_{00}$, outcome $1$ has probability zero, etc..  Thus, no matter what  result obtains, it rules out one of the four preparations.

Now suppose that we impose the Preparation Independence Postulate: the ontic state space $\Lambda$ of the joint system $AB$ is the Cartesian product of the ontic state spaces $\Lambda_A$, $\Lambda_B$  of the  component systems, and  the preparation distributions factorize; that is, there exist probability measures $P_x^A$, $P_y^B$, for
 $x, y \in \{0, + \}$, such that for all measurable $\Delta_A \subseteq \Lambda_A$ and $\Delta_B \subseteq \Lambda_B$,
\begin{equation}
P_{xy}(\Delta_A \times \Delta_B) = P_x^A(\Delta_A) P_y^A(\Delta_B).
\end{equation}
This, together with the condition that each outcome has probability zero in some preparation, yields the conclusion that  the overlaps $\omega(P_0^A, P_+^A)$ and $\omega(P_0^B, P_+^B)$ satisfy
\begin{equation}
\omega(P_0^A, P_+^A) \: \omega(P_0^B, P_+^B) = 0.
\end{equation}
If, now, $A$ and $B$ are systems of the same kind subject to the same choices of preparation procedures, $\omega(P_0^A, P_+^A) = \omega(P_0^B, P_+^B)$, and so
\begin{equation}
\omega(P_0^A, P_+^A) = \omega(P_0^B, P_+^B) = 0,
\end{equation}
and hence
\begin{equation}
\delta(P_0^A, P_+^A) = \delta(P_0^B, P_+^B) = 1.
\end{equation}
The argument generalizes beyond the idealized case of perfect preclusion.  Suppose that, for each outcome $k \in \{1, 2, 3, 4\}$, there exists a preparation $\langle x, y \rangle$ such that $P_{xy}(k) \leq \varepsilon$.
On the assumption of the Preparation Independence Postulate, it follows that
\begin{equation}
\omega(P_0^A, P_+^A) \, \omega(P_0^B, P_+^B) \leq 4 \varepsilon,
\end{equation}
and hence, if $\omega(P_0^A, P_+^A) = \omega(P_0^B, P_+^B)$,
\begin{equation}
\delta(P_0^A, P_+^A) \geq 1 - 2 \sqrt{\varepsilon}.
\end{equation}
See  ref. \cite{PBR} Supplementary Information  for details of the proof.

\section{The Preparation Uninformativeness Condition}  We wish to find a substitute for the Preparation Independence Postulate that does not presuppose the Cartesian Product Assumption, or even that the Cartesian product  of the subsystem state spaces can be embedded in the state space of the joint system. Thus, the weakenings of the CPA proposed in refs \cite{ESSV,Mansfield2016}, on which the state space considered is the Cartesian product of the component state spaces and another space that encodes nonlocal information, will not suit our purposes.  The condition that we will impose is the following.  Suppose that, for systems $A$, $B$, we have some set of possible preparations of the individual systems.  Suppose that the choice of preparation for each of the subsystems is made independently.  Following the preparation of the joint system, which consists of individual preparations on the subsystems, you are not told which preparation  has been  performed, but you are  given a complete specification of the ontic state of the joint system.  On the basis of this information, you have degrees of belief about which preparations were performed. In the case of ontologically distinct preparations, you will be certain about what preparations were performed; if the preparations are not ontologically distinct, you may have less than total information about which preparations were performed.

We ask: under these conditions, is information about which preparation was performed on one system  informative about which preparation was performed on the other? The Preparation Uninformativeness Condition is the condition that it is not. Note that it is automatically satisfied if the preparations are ontologically distinct. In such a case, given the ontic state of the joint system, you know precisely which preparations have been performed, and being told about the preparation on one system does not add to your stock of knowledge.

One way for the PUC to be violated is to have the ontic state space of the joint system  be the Cartesian product of the subsystem ontic spaces, and  the joint probability distributions  be ones in which the states of the subsystems are correlated.  It is also violated, as we shall see, by models, such as those constructed in ref. \cite{ABCL}, on which nonorthogonal quantum states are never ontologically distinct.

The PUC is implied by the PIP, but it is strictly weaker.  As shown in the Appendix, even if the CPA is assumed, it is possible to construct models for the PBR setup, in which the PUC is satisfied but $P_{00}$ and $P_{0+}$ have nonnull overlap, which by the PBR theorem, is ruled out for models that satisfy the NCA.

We now state the PUC more formally.  Let $x \in X$, $y \in Y$ be variables ranging over finite sets of possible preparations of $A$, $B$, respectively.  We assume that $x$ and $y$ are chosen independently, with  probabilities $p_x$, $q_y$ in $(0, 1)$.  Let $P_{xy}$ be the probability distribution over the ontic state space $\Lambda$ of the joint system,  corresponding to preparation $x$ on $A$ and $y$ on $B$.  The procedure of independently choosing $x$ and $y$ and preparing $P_{xy}$ yields a mixed distribution,
\begin{equation}
\bar{P} = \sum_{x, y} p_x \, q_y  \, P_{xy}.
\end{equation}
Suppose that, knowing that the system $AB$ has been prepared in this way, but not knowing which preparation was chosen,  you are given the information that the ontic state $\lambda$ of $AB$ is in some measurable subset $\Delta \subseteq \Lambda$.   Conditional on that information, your degree of belief  in the proposition $A_a$, that $A$ was subjected to preparation $a$, should be
\begin{equation}
\bar{P}(A_a | \Delta) = \frac{p_a \sum_y q_y P_{a y}(\Delta)}{\bar{P}(\Delta)}.
\end{equation}
Your degree of belief that $B$ was subjected to preparation $b$ should be
\begin{equation}
\bar{P}(B_b | \Delta) = \frac{q_b \sum_x p_x P_{x b}(\Delta)}{\bar{P}(\Delta)},
\end{equation}
and your degree of belief that the joint preparation is $\langle a, b \rangle$ should be
\begin{equation}
\bar{P}(A_a \, B_b | \Delta) = \frac{p_a \, q_b \, P_{ab}(\Delta)}{\bar{P}(\Delta)}.
\end{equation}
The condition that learning $B_b$ have no effect on your degree of belief in $A_a$, and \emph{vice versa}, is the condition that
\begin{equation}\label{ind1}
\bar{P}(A_a \, B_b | \Delta) = \bar{P}(A_a | \Delta)\bar{P}(B_b | \Delta).
\end{equation}
We  impose the condition that this hold, not for arbitrary sets $\Delta$, but for  a maximal specification of the ontic state. If the ontic state space has atoms that are assigned nonzero probability by some or all of the preparation distribution, we will take (\ref{ind1}) to hold for those; in the general case, we take it to hold in the limit of increasingly specific state descriptions.  Let the distributions $P_{x y}$ have density  functions $\mu_{xy}$ with respect to some measure $Q$ that dominates them all.  For any $\lambda \in \Lambda_{AB}$ (except, perhaps, for  a set that has  measure zero on all of the preparation distributions), and for any $\varepsilon > 0$, it is  possible to find a set $\delta_\varepsilon(\lambda)$ containing $\lambda$, with $Q(\delta_\varepsilon(\lambda)) > 0$,   on which  all of the density functions $\mu_{xy}$ vary by less than $\varepsilon$. If $\{\delta_\varepsilon(\lambda) \}$ is a family of such sets, we  have
\begin{equation}\label{limit}
\lim_{\varepsilon \rightarrow 0} \frac{P_{x y}(\delta_\varepsilon(\lambda))}{Q(\delta_\varepsilon(\lambda))} = \mu_{xy}(\lambda).
\end{equation}
We impose the condition that
\begin{equation}\label{ind1}
\lim_{\varepsilon \rightarrow 0} \left[\bar{P}(A_a \, B_b | \delta_\varepsilon(\lambda)) - \bar{P}(A_a | \delta_\varepsilon(\lambda))\,\bar{P}(B_b | \delta_\varepsilon(\lambda)) \right] = 0.
\end{equation}
This is equivalent to
\begin{equation}
\frac{\mu_{ab}(\lambda)}{\sum_{x,y} p_x \, q_y \, \mu_{x y}(\lambda)} = \frac{\sum_{x, y} p_x \, q_y \, \mu_{a y}(\lambda) \, \mu_{x b} (\lambda)}{(\sum_{x,y} p_x \, q_y \, \mu_{x y}(\lambda))^2},
\end{equation}
or,
\begin{equation}
\sum_{x, y} p_x \, q_y \left( \mu_{a b} (\lambda) \mu_{x y}(\lambda) -  \mu_{a y} (\lambda) \mu_{x b}(\lambda) \right) = 0.
\end{equation}
If this is to hold for arbitrary probabilities $\{p_x\}$, $\{q_y\}$, we must have
\begin{equation}\label{PUC}
\mu_{a b}(\lambda) \,\mu_{x y}(\lambda) = \mu_{x b} (\lambda) \, \mu_{a y}(\lambda).
\end{equation}
for all $x \in X$, $y \in Y$.   If there are  density functions  such that   (\ref{PUC}) holds for almost all $\lambda$, we can always construct density functions for which it holds for \emph{all} $\lambda$.  Though the values of the density functions will depend on the choice of background measure $Q$ used to define the densities, the choice will not affect whether (\ref{PUC}) is satisfied.

We take  as our formal definition of the Preparation Uninformativeness Condition: A set of probability distributions $\{P_{x y} \, | \,  x \in X, y \in Y \}$ on a common measurable space $\langle \Lambda, \mc{S} \rangle$ will be said to satisfy the PUC if and only if they have densities $\{ \mu_{x y} \}$ that satisfy (\ref{PUC}) for all  $\lambda \in \Lambda_{AB}$ and all $a, x \in X$, $b, y \in Y$.

For the special case in which the choice of preparations consists of only two alternatives $\{0, +\}$ for each subsystem, this is the condition that, for all $\lambda \in \Lambda_{AB}$,
\begin{equation}\label{PUC2}
\mu_{00}(\lambda) \, \mu_{++}(\lambda) = \mu_{0+}(\lambda) \, \mu_{+0}(\lambda).
\end{equation}

\section{A $\psi$-ontology result}  With the PUC in place, it is possible to show that, for the PBR example, respecting the quantum preclusions entails that $\ket{0}_A \ket{0}_B$ and $\ket{+}_A \ket{+}_B$ are ontologically distinct, as are $\ket{0}_A \ket{+}_B$ and $\ket{+}_A \ket{0}_B$.

\begin{theorem} Let $P_{x y}$, $x, y \in \{0, + \}$ be four probability distributions on a measurable space $\langle \Lambda, \mc{S} \rangle$, that satisfy the Preparation Uninformativeness Condition. If there is an experiment $E$, such that each outcome of $E$ is precluded by some $P_{x y}$, then $P_{00}$ and $P_{++}$ have null overlap, as do  $P_{0+}$ and $P_{+0}$.
\end{theorem}
\begin{proof}
Since every preparation is ruled out  by some outcome, the four distributions have null joint overlap.  This means that, for almost all $\lambda$,
\begin{equation}
\mu_{++}(\lambda) \, \mu_{00}(\lambda) \, \mu_{+0}(\lambda) \, \mu_{0+}(\lambda) = 0.
\end{equation}
We assume the PUC,
\begin{equation}
\mu_{00}(\lambda) \, \mu_{++}(\lambda) = \mu_{0+}(\lambda) \, \mu_{+0}(\lambda).
\end{equation}
Then
\begin{equation}
\mu_{++}(\lambda) \, \mu_{00}(\lambda) \, \mu_{+0}(\lambda) \, \mu_{0+}(\lambda) = (\mu_{00}(\lambda) \, \mu_{++}(\lambda))^2,
\end{equation}
and so
\begin{equation}
\mu_{++}(\lambda) \, \mu_{00}(\lambda) = \mu_{+0}(\lambda) \, \mu_{0+}(\lambda) = 0
\end{equation}
for almost all $\lambda$.
\end{proof}

It is easy to see intuitively why the PUC entails that  $P_{++}$ and $P_{00}$ have null overlap.  Since every preparation is ruled out  by some outcome, there is no subset of $\Lambda$ that is assigned positive probability by all four preparation measures.  It follows from this that, if $\mu_{++}$ and $\mu_{00}$ have non-null overlap, then the PUC is violated. To see this, suppose that, for some $\lambda$,  $\mu_{++}(\lambda)$ and $\mu_{00}(\lambda)$ are both nonzero.  Then, for any such $\lambda$ (except perhaps for a set of measure zero), either $\mu_{+0}(\lambda)$ or $\mu_{0+}(\lambda)$ (or both) is zero. Suppose that $\mu_{+0}(\lambda) = 0$.  Then, conditional on the information that the state is $\lambda$, you have nonzero degree of belief in preparations $\ket{0}_A \ket{0}_B$ and $\ket{+}_A \ket{+}_B$, and zero degree of belief
 in $\ket{+}_A \ket{0}_B$.  If, now, you are informed that the $A$-preparation was $\ket{+}_A$, this permits you to  conclude that the $B$-preparation was $\ket{+}_B$, in violation of the PUC.

The assumption that strict preclusion is possible is not an  assumption that survives the transition to quantum field theory. In a quantum field theory, no experiment that can be performed in a bounded spacetime region has probability strictly zero for any outcome.  We need, therefore, a generalization of the result that does not presume strict preclusion.

We will say that an outcome $k$ is $\varepsilon$-precluded by a probability distribution $P$ iff the probability assigned to $k$ by $P$ is less than or equal to $\varepsilon$.
In place of strict preclusion, we will assume that, for every $\varepsilon > 0$, there is an experiment that achieves $\varepsilon$-preclusion.  This is sufficient to yield the conclusion that $P_{++}$ and $P_{00}$ are ontologically distinct.

\begin{theorem}\label{bounds} Let $P_{x y}$, $x, y \in \{0, + \}$ be four probability distributions on a measurable space $\langle \Lambda, \mc{S} \rangle$ that satisfy the Preparation Uninformativeness Condition.  If there is an $m$-outcome experiment $E$, such that each outcome of $E$ is $\varepsilon$-precluded by  some $P_{x y}$, then
\begin{equation*}
\delta(P_{00}, P_{++}) \geq H(P_{00}, P_{++})^2 \geq 1 - 2 \sqrt{m \varepsilon},
\end{equation*}
and
\begin{equation*}
\delta(P_{0+}, P_{+0}) \geq H(P_{0+}, P_{+0})^2 \geq 1 - 2 \sqrt{m \varepsilon}.
\end{equation*}
\end{theorem}
\begin{proof}
Let $\bar{P}$ be the equally weighted mixture of the four probability distributions.
\begin{equation}
\bar{P} = \frac{1}{4} \sum_{x,y} P_{x y}.
\end{equation}
Since each $P_{xy}$ is absolutely continuous with respect to $\bar{P}$, by the Radon-Nikodym theorem  it has a density $\mu_{x y}$ with respect to $\bar{P}$.  Since, for each $x, y$, $\bar{P}(\Delta) \geq P_{xy}(\Delta)/4$ for all $\Delta \in \mc{S}$,  we must have
\begin{equation}
\mu_{x y} (\lambda) \leq 4
\end{equation}
for almost all $\lambda \in \Lambda$.

Let $K$ be the outcome-set of the experiment $E$, and let $\{p_k \, | \, k \in K \}$ be the functions that yield probabilities for these outcomes.  These must satisfy, for all $\lambda$,
\begin{equation}
\sum_{k \in K} p_k(\lambda) = 1.
\end{equation}

Let $K_0$ be the subset of $K$ consisting of outcomes that are $\varepsilon$-precluded by either $P_{00}$ or $P_{++}$, and let $K_1$ be the set of  outcomes that are $\varepsilon$-precluded by either $P_{0+}$ or $P_{+0}$.  Since every outcome is $\varepsilon$-precluded by some preparation, $K_0 \cup K_1 = K$.

If outcome $k$ is $\varepsilon$-precluded by $P_{00}$, then, because $\mu_{++}(\lambda) \leq 4$ for almost all $\lambda$,
\begin{multline}
\int_\Lambda \mu_{00}(\lambda) \mu_{++}(\lambda)\, p_k(\lambda) \, d\bar{P}(\lambda) \\
\leq  4 \int_\Lambda \mu_{00}(\lambda)\, p_k(\lambda) \, d\bar{P}(\lambda) \leq 4 \varepsilon.
\end{multline}
Similarly, if $k$ is $\varepsilon$-precluded by $P_{++}$,
\begin{equation}
\int_\Lambda \mu_{00}(\lambda) \mu_{++}(\lambda) \, p_k(\lambda) \, d\bar{P}(\lambda) \leq 4 \varepsilon.
\end{equation}
Therefore, for all $k \in K_0$,
\begin{equation}
\int_\Lambda \mu_{00}(\lambda) \mu_{++}(\lambda) \, p_k(\lambda) \, d\bar{P}(\lambda) \leq 4 \varepsilon.
\end{equation}
Similarly, for all $k \in K_1$,
\begin{equation}
\int_\Lambda \mu_{0+}(\lambda) \mu_{+ 0}(\lambda) \, p_k(\lambda) \, d\bar{P}(\lambda) \leq 4 \varepsilon.
\end{equation}
If, now, we assume the PUC,
\begin{equation}
\mu_{00}(\lambda) \mu_{++}(\lambda) = \mu_{0+}(\lambda) \mu_{+ 0}(\lambda),
\end{equation}
it follows that, for all $k \in K$,
\begin{multline}
\int_\Lambda \mu_{00}(\lambda) \mu_{++}(\lambda) \, p_k(\lambda)\, d\bar{P}(\lambda) \\
= \int_\Lambda \mu_{0+}(\lambda) \mu_{+ 0}(\lambda) \, p_k(\lambda)\, d\bar{P}(\lambda) \leq 4 \varepsilon.
\end{multline}
Summing over all $k \in K$ yields
\begin{multline}
\int_\Lambda \mu_{00}(\lambda) \mu_{++}(\lambda) \, d\bar{P}(\lambda) \\
= \int_\Lambda \mu_{0+}(\lambda) \mu_{+ 0}(\lambda) \, d\bar{P}(\lambda) \leq 4 m \varepsilon.
\end{multline}
Because of the inequalities (\ref{lapineq}),
\begin{equation}
\begin{array}{l}
\omega(P_{00}, P_{++}) \leq F(P_{00}, P_{++}) \leq  2 \sqrt{m \varepsilon}
\\ \\
\omega(P_{0+}, P_{+0}) \leq F(P_{0+}, P_{+0}) \leq  2 \sqrt{m \varepsilon},
\end{array}
\end{equation}
from which it follows that
\begin{equation}
\begin{array}{l}
\delta(P_{00}, P_{++}) \geq H(P_{00}, P_{++})^2 \geq   1 - 2 \sqrt{m \varepsilon},
\\ \\
\delta(P_{0+}, P_{+0}) \geq H(P_{0+}, P_{+0})^2 \geq 1  -  2 \sqrt{m \varepsilon}.
\end{array}
\end{equation}
\end{proof}

\begin{corollary} Let $P_{x y}$, $x, y \in \{0, + \}$ be four probability distributions on a measurable space $\langle \Lambda, \mc{S} \rangle$ that satisfy the Preparation Uninformativeness Condition.  If, for every $\varepsilon > 0$,  there is an experiment $E_\varepsilon$, such that each outcome of $E_\varepsilon$ is $\varepsilon$-precluded by  some $P_{x y}$, then
\[
\delta(P_{00}, P_{++}) = \delta(P_{0+}, P_{+0}) = 1.
\]
\end{corollary}

The condition that $P_{00}$ be ontologically distinct from $P_{++}$ and  $P_{0+}$ be ontologically distinct from $P_{+0}$ entails that, though the ontic state $\lambda$ might not uniquely determine which preparation was performed, it always uniquely determines at least \emph{one} of the two preparations.

If an ontic state $\lambda$ is compatible with $\ket{0}_A \ket{0}_B$ but does not uniquely determine the $A$-preparation, then, since it can't be compatible with $\ket{+}_A \ket{+}_B$, it must be compatible with  $\ket{+}_A \ket{0}_B$  and, hence, incompatible with $\ket{0}_A \ket{+}_B$.  In such a case, $\lambda$ is compatible with only one $B$-preparation,  $\ket{0}_B$.  Similar conclusions hold for any other ontic states; any ontic state will either uniquely determine the $A$-preparation or the $B$-preparation.  Thus, if the PUC is satisfied in a situation in which two choices are made for $A$-preparation and for $B$-preparation,  for almost all ontic states $\lambda$, $\lambda$ determines either a fact of the matter about the quantum state of $A$ or a fact of the matter about the quantum state of $B$.

Now consider the case of a large number $N$ of systems, each of which is subject to a $\ket{0}$ or $\ket{+}$ preparation.  This gives us $2^N$ possible preparations of the joint system.
Suppose that, for each pair $\langle \alpha, \beta  \rangle$ of these systems, there is an experiment such that each outcome of the experiment rules out one of  $\{ \ket{0}_\alpha \ket{0}_\beta, \ket{0}_\alpha \ket{+}_\beta, \ket{+}_\alpha \ket{0}_\beta, \ket{+}_\alpha \ket{+}_\beta \}$. That is, for each $\langle x, y \rangle, x, y \in \{ 0, +\}$, there is an outcome that is assigned zero probability by all of the $2^{N-2}$ probability distributions corresponding to preparations in which system $\alpha$ is subjected to the $\ket{x}$ preparation and system $\beta$ to the $\ket{y}$  preparation.  Suppose that the PUC is satisfied.  By the same reasoning we have applied to the case of only two systems, it must be the case that, for any pair $\alpha, \beta$ of systems, for almost all ontic states $\lambda$,  the ontic state $\lambda$ uniquely determines the quantum state of at least one member of the pair.  Since this must be true for every pair, at most one of the systems has its quantum state undetermined  by the ontic state $\lambda$.

This means that, for large $N$, for almost all $\lambda$, most of the systems---at least  $N - 1$ of them---will have their quantum states uniquely determined by the ontic state $\lambda$.  If one of these systems is chosen at random, the probability of choosing a system in  a definite quantum state is at least $1 - 1/N$.

We now impose a \emph{Principle of Extendibility}: any system composed of $N$ subsystems of the same type can be regarded as a part of a larger system consisting of a greater number of systems of the same type.  If we assume this, then, the state of the individual systems must be such that, if one is chosen at random, the probability that its quantum state is uniquely determined by the ontic state is, for every $N$, greater than or equal to $1 - 1/N$.    That is, with probability one, the ontic state of any whole of which the system may be regarded as a part is such that the ontic state uniquely determines its quantum state.

Therefore, if we assume the PUC and the Principle of Extendibility, and assume also  the possibility of experiments that achieve $\varepsilon$-preclusion for arbitrarily small $\varepsilon$, then, for any system, the preparations $\ket{0}$, $\ket{+}$ are ontologically distinct.

As is shown in ref.\cite{MoseleyPBR}, for any pair $\{ \ket{\psi}, \ket{\phi} \}$ such that $|\bkt{\phi}{\psi}| \leq 1/\sqrt{2}$, there is a 4-outcome experiment such that each outcome is precluded by one of $\{\ket{\psi}\ket{\psi}, \ket{\psi}\ket{\phi}, \ket{\phi}\ket{\psi}, \ket{\phi}\ket{\phi} \}$, and hence our proof includes such cases.  The extension in ref. \cite{PBR} of the  result to arbitrary distinct states makes essential use of the Cartesian Product Assumption.  There does not seem to be an obvious way to use the PUC to demonstrate ontological distinctness of arbitrary distinct pure quantum states.  We leave it as an open question whether this is possible. The argument of the present paper can, however, be invoked to show that, for the case of $\{ \ket{\psi}, \ket{\phi} \}$ with $ 1/\sqrt{2} < |\bkt{\phi}{\psi}| < 1$,  for sufficiently large $N$ the $N$-fold product $\ket{\psi}_1 \ldots \ket{\psi}_N$ is ontologically distinct from $\ket{\phi}_1 \ldots \ket{\phi}_N$.

\section{Relaxing the assumption of arbitrarily strict preclusion} In the preceding section we assumed that arbitrarily close approximations to strict preclusion are possible. This assumption holds in quantum theories, including quantum field theories.  If, however, we do not assume that a putative successor theory will reproduce the quantum predictions precisely, but require only a close approximation to quantum predictions in the domains in which quantum theories are known to work well, there may be in such a theory a bound on how close an approximation to strict preclusion can be achieved.  One would expect that, if this bound is sufficiently small, distinct quantum states must be at least approximately ontologically distinct.  In this section we show that this is the case.

To do this, we need a means of quantifying approximate ontological distinctness.  To do this, we imagine the following game.  We have a system $\Sigma_N$ composed of  $N$ subsystems, each of which is to be subjected to a $\ket{0}$ or $\ket{+}$ preparation, chosen at random, with each preparation equiprobable.  For each preparation, there is a corresponding probability distibution over the ontic state space of $\Sigma_N$; it is assumed that these are known to you.   A preparation of this sort having been performed, you are presented with the ontic state $\lambda$ of the joint system.  A referee picks one of the subsystems at random, and you are asked to guess its preparation, based on knowledge of the ontic state $\lambda$ of the joint system. What is the probability that you will guess correctly?  If there is a strategy such that, with probability one, your guess is correct, this warrants the conclusion that, with probability one, the ontic state $\lambda$ of the whole uniquely determines the quantum state of each subsystem. If there is a strategy such that, with probability close to one, your guess is correct, this warrants the conclusion that, with high probability, the ontic state $\lambda$ of the whole is such that it at least approximately determines the quantum state of most of the subsystems.

In the case of perfect preclusion, on the assumption of the PUC and the Principle of Extendibility, there is a strategy that succeeds with probability one. As we have seen in the previous section, if, for each pair of subsystems $\langle \alpha, \beta \rangle$, there is an experiment such that  each outcome of the experiment rules out one of  $\{ \ket{0}_\alpha \ket{0}_\beta, \ket{0}_\alpha \ket{+}_\beta, \ket{+}_\alpha \ket{0}_\beta, \ket{+}_\alpha \ket{+}_\beta \}$, then this, together with the PUC, entails that, for almost all ontic states $\lambda$, for at least $N - 1$ of the subsystems the quantum state is uniquely determined by the ontic state.  There is, therefore, a strategy that is guaranteed to get at least $N - 1$ guesses right.  The remaining system may have its quantum state left undetermined by the ontic state.  In the absence of any information about the ontic state, a random guess has probability one-half of being correct.  Knowledge of the ontic state cannot decrease the probability on an optimal strategy of guessing  the state of any subsystem correctly.  Therefore, on an optimal strategy for this game, the expectation value of the number of correct guesses must satisfy,
\begin{equation}
\langle \mbox{\# correct guesses} \rangle \geq N - \frac{1}{2}.
\end{equation}
If one of the subsystems is chosen at random, the probability that you have correctly guessed its preparation is
\begin{multline}\label{probN}
P(\mbox{correct guess}) = \frac{1}{N} \langle \mbox{\# correct guesses} \rangle
\\
\geq 1 - \frac{1}{2N}.
\end{multline}
We now assume the Principle of Extendibility, which requires that the state of our composite system be compatible with its being regarded as part of an larger system composed of an arbitrarily high number of subsystems.  On this assumption, inequality (\ref{probN}) must hold for \emph{all} $N$.  Therefore, the probability of correcting guessing the state of a randomly chosen subsystem is
\begin{equation}
P(\mbox{correct guess}) = 1,
\end{equation}
from which it follows that the ontic state of the whole  uniquely determines the quantum state of each of the subsystems.

We now show that, if it is possible to achieve $\varepsilon$-preclusion, for small $\varepsilon$, then the probability of guessing correctly is close to one.

We assume pairwise $\varepsilon$-preclusion.  That is, we assume that there exists $\varepsilon > 0$ such that, for each pair of subsystems $\langle \alpha, \beta \rangle$, there is a $4$-outcome experiment $E^{\alpha \beta}$  such that, for  each outcome $k$  of  $E^{\alpha \beta}$, there are $x, y \in \{0, +  \}$ such that $k$ has probability less than or equal to $\varepsilon$ on any preparation in which $\alpha$ is subjected to $\ket{x}$ and $\beta$ to $\ket{y}$.

For any subsystem $\alpha$, for $x \in \{0, +\}$, let $\bar{P}^{\alpha}_{x}$ be the probability distribution that is an equally weighted mixture of all distributions on which $\alpha$ has been subjected to the $\ket{x}$ preparation, and let $\bar{\mu}^\alpha_x$ be its density function.    For any pair $\langle \alpha, \beta \rangle$, and any $x, y \in \{ 0, + \}$, let $\bar{P}^{\alpha \beta}_{x y}$ be the equally weighted mixture of all distributions on which $\alpha$ has been subjected to the $\ket{x}$ preparation and $\beta$ to the $\ket{y}$ preparation, and let $\bar{\mu}^{\alpha \beta}_{x y}$ be the corresponding density function. We must have, by Theorem \ref{bounds}.
\begin{equation}
\omega(\bar{P}^{\alpha \beta}_{++}, \bar{P}^{\alpha \beta}_{0 0}) \leq 2 \sqrt{\varepsilon}; \quad
\omega(\bar{P}^{\alpha \beta}_{+ 0}, \bar{P}^{\alpha \beta}_{0 +}) \leq 2 \sqrt{\varepsilon}.
\end{equation}
We adopt the following strategy.  Given an ontic state $\lambda$ of $\Sigma_N$,
\begin{itemize}
\item For any subsystem $\alpha$, guess $\ket{0}_\alpha$ if $\bar{\mu}^\alpha_{0}(\lambda) \geq  \bar{\mu}^\alpha_{+}(\lambda)$;  $\ket{+}_\alpha$ otherwise.
\end{itemize}
This has the consequence that, for any pair $\langle \alpha, \beta \rangle$,
\begin{itemize}
\item If your guess is $\ket{+}_\alpha \ket{+}_\beta$, then $\bar{\mu}^{\alpha \beta}_{++}(\lambda) >  \bar{\mu}^{\alpha \beta}_{00}(\lambda)$;
\item If your guess is $\ket{+}_\alpha \ket{0}_\beta$, then $\bar{\mu}^{\alpha \beta}_{+0}(\lambda) >  \bar{\mu}^{\alpha \beta}_{0+}(\lambda)$;
\item If your  guess is $\ket{0}_\alpha \ket{+}_\beta$, then $\bar{\mu}^{\alpha \beta}_{0+}(\lambda) >  \bar{\mu}^{\alpha \beta}_{+0}(\lambda)$;
\item If your guess is $\ket{0}_\alpha \ket{0}_\beta$, then $\bar{\mu}^{\alpha \beta}_{00}(\lambda) \geq  \bar{\mu}^{\alpha \beta}_{++}(\lambda)$.
\end{itemize}
We ask: on this strategy, for any pair $\langle \alpha, \beta \rangle$, what is the probability that the strategy yields incorrect guesses for both members of the pair?

If the preparation is one on which $\alpha$ and $\beta$ are both subjected to the  $\ket{0}$ preparation, then both guesses are incorrect only if  $\bar{\mu}^{\alpha \beta}_{++}(\lambda) >  \bar{\mu}^{\alpha \beta}_{00}(\lambda)$.  Let $\Lambda^{\alpha \beta}_{++}$ be the subset of $\Lambda$  on which this holds, and let $\Omega^{\alpha \beta}_{++}$ be its complement in $\Lambda$.  The probability that both guesses are incorrect is $\bar{P}^{\alpha \beta}_{00}(\Lambda^{\alpha \beta}_{++})$.  We have
\begin{align}
\nonumber \bar{P}^{\alpha \beta}_{00}(\Lambda^{\alpha \beta}_{++}) &= \int_{\Lambda_{++}} \bar{\mu}^{\alpha \beta}_{00}(\lambda) \, d \bar{P}(\lambda)
\\ \nonumber
&\leq \int_{\Lambda_{++}} \bar{\mu}^{\alpha \beta}_{00}(\lambda) \, d \bar{P}(\lambda) + \int_{\Omega^{\alpha \beta}_{++}} \bar{\mu}^{\alpha \beta}_{++}(\lambda) \, d \bar{P}(\lambda)
\\ \nonumber
&= \int_\Lambda \min \left(\bar{\mu}^{\alpha \beta}_{00}(\lambda),\bar{\mu}^{\alpha \beta}_{++}(\lambda)\right) \, d \bar{P}(\lambda)
\\
&= \omega(\bar{P}^{\alpha \beta}_{00}, \bar{P}^{\alpha \beta}_{++}).
\end{align}
Thus, when the preparation is $\ket{0}_\alpha\ket{0}_\beta$, the probability that both guesses are incorrect is less than or equal to $\omega(\bar{P}^{\alpha \beta}_{00}, \bar{P}^{\alpha \beta}_{++})$.  By Theorem \ref{bounds}, this is less than or equal to $2 \sqrt{\varepsilon}$.  The same holds for the other preparations.  Therefore, the probability of two incorrect guesses satisfies,
\begin{equation}
P(\mbox{both guesses incorrect}) \leq 2 \sqrt{\varepsilon}.
\end{equation}
If, out of all the guesses for all $N$ subsystems, our strategy yields  more than one incorrect guess, then for some pair $\langle \alpha, \beta \rangle$, both members of  are guessed incorrectly.  There are $N(N-1)/2$ pairs in all, for each of which the probability of getting both incorrect is at most $2 \sqrt{\varepsilon}$.  ; therefore, the probability of more than one incorrect guess satisfies,
\begin{equation}
P(\mbox{more than one incorrect}) \leq N(N-1)\sqrt{\varepsilon}.
\end{equation}
The expectation value of  number of correct guesses satisfies
\begin{equation}
\langle \mbox{\# correct guesses} \rangle \geq (N-1) (1 - N(N-1)\sqrt{\varepsilon}),
\end{equation}
and the probability that a randomly chosen subsystem is guessed correctly satisfies
\begin{multline}
P(\mbox{correct guess}) \geq (1- 1/N) (1 - N(N-1)\sqrt{\varepsilon})
\\
= 1 - 1/N - (N-1)^2 \sqrt{\varepsilon}.
\end{multline}
If, now, we assume the Principle of Extendibility, this must be true for all $N$.  We can use this to obtain a lower bound on the probability of a correct guess. We have, for all $N$,
%\begin{align}
%\nonumber P(\mbox{correct guess}) &\geq (1-1/N) (1 - N(N-1)\sqrt{\varepsilon})
%\\ \nonumber &= 1 - 1/N - (N - 1)^2 \sqrt{\varepsilon}
%\\ &\geq \ 1 - 1/N  -  N^2 \sqrt{\varepsilon}).
%\end{align}
\begin{multline}
P(\mbox{correct guess}) \geq 1 - 1/N - (N-1)^2 \sqrt{\varepsilon}
\\ \geq \ 1 - 1/N  -  N^2 \sqrt{\varepsilon}.
\end{multline}

Let $N_\varepsilon$ be the largest integer $N$ such $2 N^3 \sqrt{\varepsilon} \leq 1$.  That is,
\begin{equation}\label{epsineq}
2 N_\varepsilon^3 \sqrt{\varepsilon} \leq 1  \leq 2 (N_\varepsilon + 1)^3 \sqrt{\varepsilon}.
\end{equation}
Then, because of the left-hand inequality in (\ref{epsineq}),
\begin{equation}
N_\varepsilon^2 \sqrt{\varepsilon} \leq 2^{-2/3} \varepsilon^{1/6},
\end{equation}
and, because of the right-hand inequality in (\ref{epsineq}),
\begin{equation}
N_\varepsilon \geq 2^{-1/3} \varepsilon^{-1/6} - 1,
\end{equation}
and so
\begin{equation}
{1}/{N_\varepsilon} + N^2_\varepsilon \sqrt{\varepsilon} \leq \left(\frac{2^{1/3}}{1 - 2^{1/3}\varepsilon^{1/6}} + 2^{-2/3} \right) \varepsilon^{1/6}.
\end{equation}
This gives us our bound.
\begin{align}
P(\mbox{correct guess}) \geq 1 - \left(\frac{2^{1/3}}{1 - 2^{1/3}\varepsilon^{1/6}} + 2^{-2/3} \right) \varepsilon^{1/6}.
\end{align}
To leading order in $\varepsilon$, the right hand side of this is
\begin{equation}
1 - \left(2^{1/3} + 2^{-2/3} \right) \varepsilon^{1/6} \approx 1 - 1.9 \: \varepsilon^{1/6}.
\end{equation}

\section{Appendix} In this section we illustrate the assertion that, even if the CPA is assumed, the PUC is strictly weaker than the NCA, by constructing a simple model of the setup of PBR's example that reproduces the quantum probabilities for the outcomes of the experiment considered, and satisfies the CPA and PUC, but not the NCA.

We adapt Spekkens' toy theory for the purpose \cite{SpekkensToy}.  In Spekkens' theory, an elementary system  consists of four boxes and a ball that can be in any one of the four boxes.  On this theory, nonorthogonal preparations are ontologically indistinct, and product-state preparations on a pair of systems yield independent probabilities for the states of the systems. The four preparations in our example, therefore,  have nonzero common overlap when modelled in this theory, which is incompatible with the quantum preclusions for the experiment considered.

To construct a model that reproduces the correct quantum probabilities for the  experiment considered, we must, therefore, change the theory's rules for assigning probability distributions to preparations. The probability distributions corresponding to these preparations must be ones on which the states of the two subsystems are not probabilistically independent.  It is possible to do so while satisfying the PUC.

Here is one way to do it.  Let the $\ket{0}$ preparation be one in which the boxes $1, 2, 3$ have probabilities $1/2, 1/4, 1/4$, respectively, and the $\ket{+}$ preparation be one on which boxes $2, 3, 4$ have probabilities $1/4, 1/4, 1/2$.  Let the joint preparation distributions be as depicted in Fig \ref{Probs}. In these diagrams, the horizontal axis represents the position of ball $A$, and the vertical axis, the position of ball $B$.  The support of the probability distribution is indicated by the shaded squares, and each shaded square is taken to have equal probability.  On this choice of distributions, $P_{00}$ and $P_{++}$ are ontologically distinct, as they must be in order to recover the preclusions while satisfying the PUC, though each has non-null overlap with the other two distributions.  The same is true of $P_{0+}$ and $P_{+0}$, and so the PUC is satisfied.

To complete our model, it remains to fill in the response probabilities for the  experiment.  We must specify four probability functions $\{p_1, p_2, p_3, p_4 \}$ on the state space of the joint system, in such a way that the preparation distributions we have already specified yield the correct probabilities for the outcome of the experiment, including, crucially, the preclusions.  The simplest way to do this is to adopt deterministic response functions; we partition the state space into four disjoint regions, in each of which one of the four outcomes of the experiment is yielded with certainty.  The constraint of yielding the correct  quantum probabilities leaves considerable arbitrariness in how this is done; in Figure \ref{Response} is exhibited one way of doing it.  The outcome of the experiment is $\ket{\xi_1}$ if both balls are in  box $3$ or $4$, $\ket{\xi_2}$ if ball $A$ is in box 3 or 4 and ball $B$ is in box 1 or 2, $\ket{\xi_3}$ if ball $A$ is in box 1 or 2 and  ball $B$ is in box 3 or 4, and $\ket{\xi_4}$ if both balls are in box 1 or 2.

\begin{figure}[h]
\centering
\begin{subfigure}[b]{0.2\textwidth}{\includegraphics[width=0.8\textwidth]{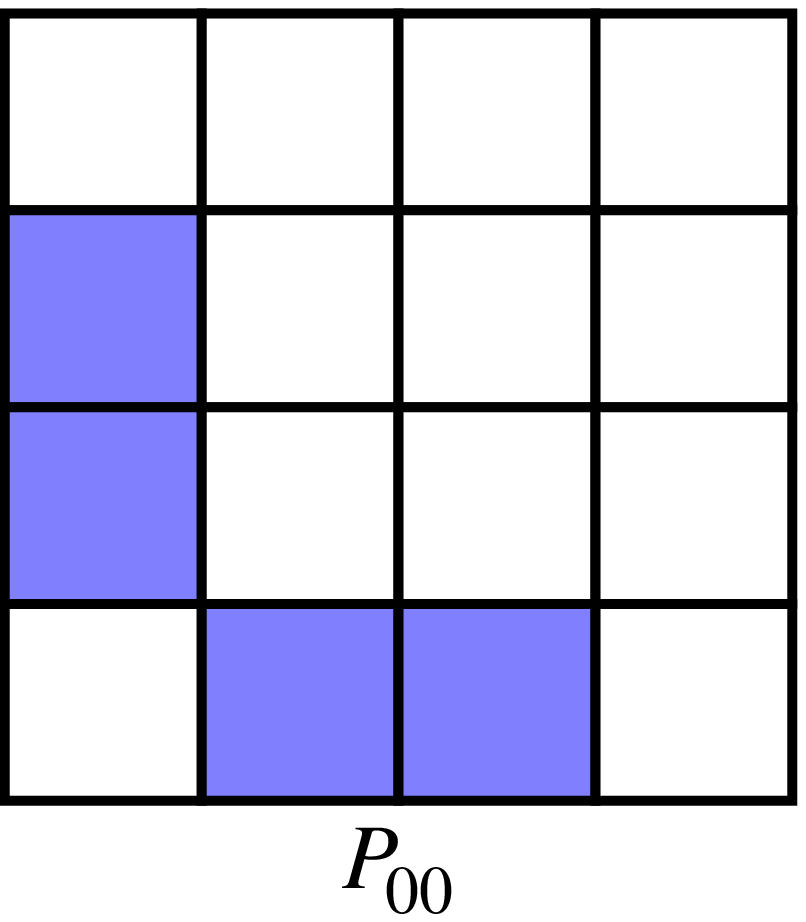}}
\end{subfigure}
\begin{subfigure}[b]{0.2\textwidth}{\includegraphics[width=0.8\textwidth]{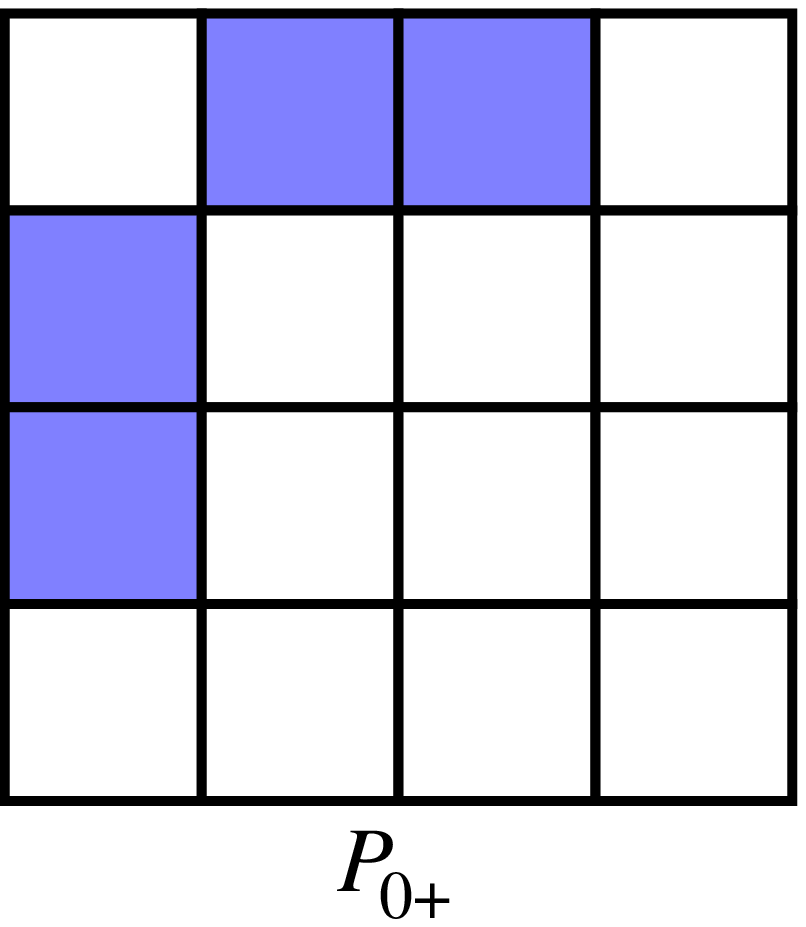}}
\end{subfigure}
\begin{subfigure}[b]{0.2\textwidth}{\includegraphics[width=0.8\textwidth]{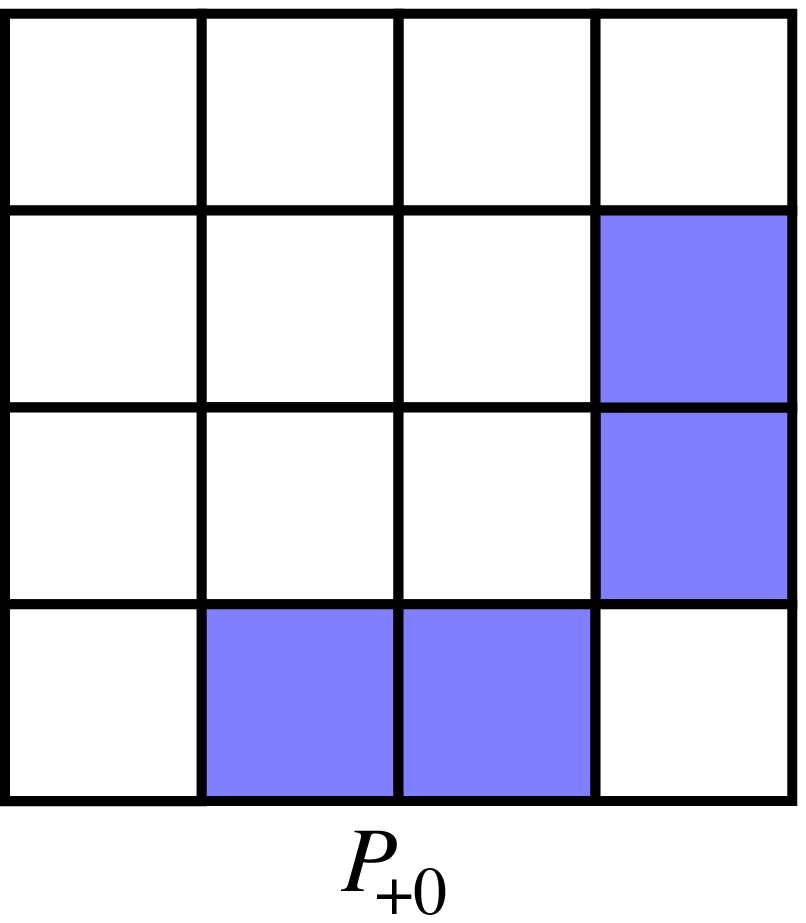}}
\end{subfigure}
\begin{subfigure}[b]{0.2\textwidth}{\includegraphics[width=0.8\textwidth]{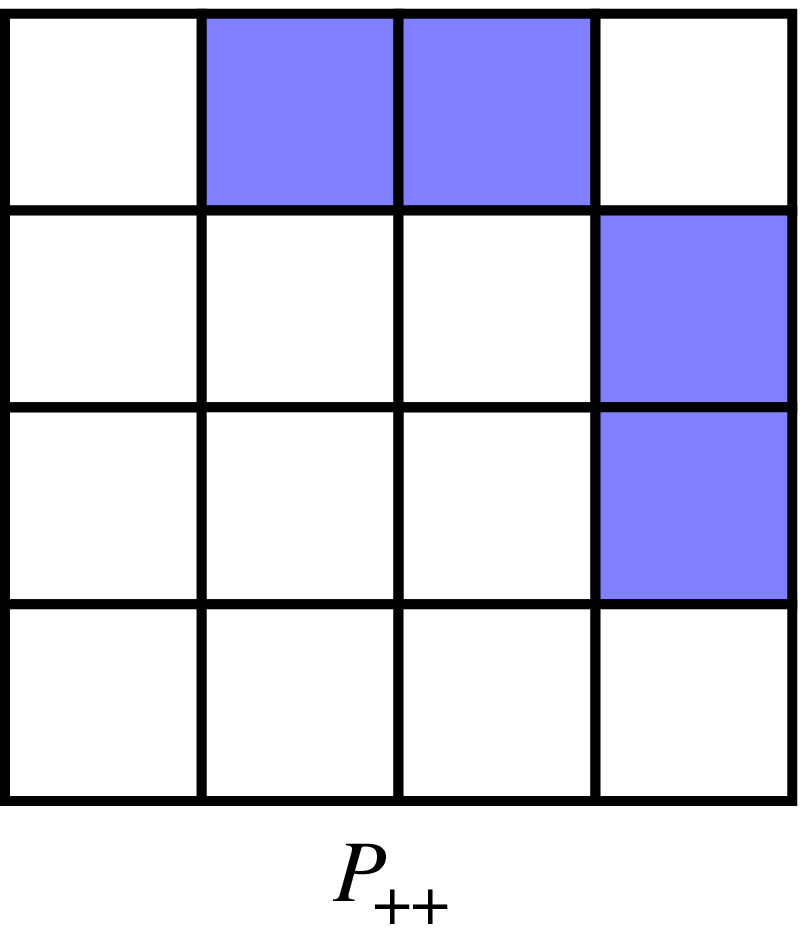}}
\end{subfigure}
\caption{The probability distributions for the example.}\label{Probs}
\end{figure}

\begin{figure}[h]
\centering
{\includegraphics[width=0.16\textwidth]{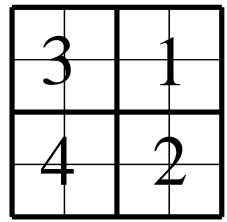}}
\caption{Response regions for the four outcomes.}\label{Response}
\end{figure}

%\bibliographystyle{apsrev4-1}
%\bibliography{QBib}
%

\end{document}